\documentclass[a4paper,USenglish,cleveref]{lipics-v2021}

\hideLIPIcs
\relatedversiondetails[cite=DBLP:conf/wdag/LeinweberH23]{Publisher Edition}{https://doi.org/10.4230/LIPIcs.DISC.2023.43}

\EventEditors{Rotem Oshman}
\EventNoEds{1}
\EventLongTitle{37th International Symposium on Distributed Computing (DISC 2023)}
\EventShortTitle{DISC 2023}
\EventAcronym{DISC}
\EventYear{2023}
\EventDate{October 10-12, 2023}
\EventLocation{L'Aquila, Italy}
\EventLogo{}
\SeriesVolume{281}
\ArticleNo{40}

\pdfoutput=1 
\nolinenumbers 
\category{Brief Announcement}

\usepackage[ruled]{algorithm}
\usepackage[noend]{algpseudocode}
\algdef{SL}[REPLICASTATE]{ReplicaState}{0}[1]{\textbf{state} #1}
\algblockdefx[FN]{Fn}{EndFn}%
[3]{\textbf{function} \emph{#1}(#2) : #3}{}
\algblockdefx[PROC]{Proc}{EndProc}%
[2]{\textbf{function} \emph{#1}(#2)}{}
\algblockdefx[UPON]{Upon}{EndUpon}%
[2]{\textbf{upon} \texttt{#1} (#2)}{}
\algtext*{EndUpon}
\algtext*{EndFn}
\algtext*{EndProc}
\newcommand{\mfn}[2]{\mathrm{#1}(#2)}
\newcommand{\quorum}{\lfloor \frac{n}{2} \rfloor + 1}

\title{Let It TEE: Asynchronous Byzantine Atomic Broadcast with $n \geq 2f+1$}

\titlerunning{Let It TEE}

\author{Marc Leinweber}{Institute of Information Security and Dependability (KASTEL), Karlsruhe Institute of Technology (KIT), Germany}{marc.leinweber@kit.edu}{https://orcid.org/0000-0002-9638-8526}{}
\author{Hannes Hartenstein}{Institute of Information Security and Dependability (KASTEL), Karlsruhe Institute of Technology (KIT), Germany}{hannes.hartenstein@kit.edu}{https://orcid.org/0000-0003-3441-3180}{}

\authorrunning{M. Leinweber and H. Hartenstein}

\Copyright{Marc Leinweber and Hannes Hartenstein}

\ccsdesc{Security and privacy~Distributed systems security}

\keywords{Byzantine Fault Tolerance, Trusted Execution Environments, Asynchrony}

\funding{This work was supported by funding from the topic Engineering Secure Systems of the Helmholtz Association (HGF).}

\begin{document}

\maketitle  

\begin{abstract}
Asynchronous Byzantine Atomic Broadcast (ABAB) promises simplicity in implementation as well as increased performance and robustness in comparison to partially synchronous approaches.
We adapt the recently proposed DAG-Rider approach to achieve ABAB with $n\geq 2f+1$ processes, of which $f$ are faulty, with only a constant increase in message size. We leverage a small Trusted Execution Environment (TEE) that provides a unique sequential identifier generator (USIG) to implement Reliable Broadcast with $n>f$ processes and show that the quorum-critical proofs still hold when adapting the quorum size to $\quorum$. This first USIG-based ABAB preserves the simplicity of DAG-Rider and serves as starting point for further research on TEE-based ABAB.
\end{abstract}


\section{Introduction}

Atomic Broadcast primitives play a crucial role for Byzantine-fault tolerant (BFT) State Machine Replication (SMR).
A prominent example for BFT SMR in the partially synchronous model is PBFT \cite{DBLP:journals/tocs/CastroL02}.
By use of small Trusted Execution Environments (TEE) that generate and sign unique sequential identifiers on each process, called USIGs, 
Veronese et al.~\cite{DBLP:journals/tc/VeroneseCBLV13} showed that PBFT's communication complexity can be reduced and the fault tolerance can be increased to $n \geq 2f+1$ while still tolerating Byzantine faults.
The authenticity/integrity of TEEs can be verified remotely and, thus, TEEs are assumed to only fail by crashing.
However, as shown by Miller et al.~\cite{DBLP:conf/ccs/MillerXCSS16}, Asynchronous Byzantine Atomic Broadcast (ABAB) outperforms approaches based on the partially synchronous model particularly under faults and tends to show a simpler design.
While it is known that any asynchronous crash fault-tolerant algorithm can be compiled to withstand Byzantine faults using TEEs \cite{DBLP:conf/podc/Ben-DavidCS22,DBLP:conf/podc/ClementJKR12}, 
the proposed compilers show either a polynomial or an exponential overhead in runtime. 
We are interested in a simple and straightforward design of a USIG-enhanced ABAB that does not add further message rounds and only adds a constant number of bits to each message (essentially a counter value and a signature).
To this end, we adapt DAG-Rider \cite{DBLP:conf/podc/KeidarKNS21} to provide ABAB with $n \geq 2f+1$ processes.
We give a quick recap on DAG-Rider and explain the adaption \emph{TEE-Rider}. 
Besides using TEE-based Reliable Broadcast and changing the required quorums from $2f+1$ to $\quorum$, we leave DAG-Rider unchanged.
We show that the quorum-based arguments of DAG-Rider still hold for TEE-Rider.

\section{TEE-Rider: Transforming DAG-Rider to $n\geq2f+1$}
\label{sec:tee-rider}

We make use of the following definition of Atomic Broadcast for a set of processes $P, n := |P|$. 
The processes communicate over authenticated point-to-point links with eventual delivery. 

\begin{definition}[Atomic Broadcast]
\label{def:abc}
    Each process $p_i \in P$ receives client transactions $t$ via events $\mfn{clientRequest}{t}$.
    Correct processes deliver tuples $(t, r, p_i)$, where $t$ is a client transaction, $r \in \mathbb{N}_0$ a round number, and $p_i \in P$ the process that initially received $t$, satisfying the following properties: \\
    \textbf{Agreement:} If a correct process $p_i \in P$ delivers $(t, r, p_j)$, then every other correct process $p_k \in P, k \not = i$ eventually delivers $(t, r, p_j)$ with probability 1.\\
    \textbf{Integrity:} For each round $r \in \mathbb{N}_0$ and process $p_j \in P$, a correct process $p_i \in P$ delivers $(t, r, p_j)$ at most once.\\
    \textbf{Validity:} If a correct process $p_i \in P$ receives an event $\mfn{clientRequest}{t}$, then every correct process $p_k \in P$ eventually delivers $(t, r, p_i)$ with probability 1.\\
    \textbf{Total Order:} Let $m_1$ and $m_2$ be any two valid tuples that are delivered by any two correct processes $p_i, p_j \in P$. If $p_i$ delivers $m_1$ before $m_2$, then $p_j$ delivers $m_1$ before $m_2$.
\end{definition}

\subsection{Changes in Assumptions, Building Blocks, and Setup} 
In addition to the assumptions of DAG-Rider, we assume that each process is equipped with a USIG \cite{DBLP:journals/tc/VeroneseCBLV13} that may only fail by crashing. It implements a signature service that binds a unique counter value to each signature it produces.
The USIG is used for Reliable Broadcast with a fault tolerance of $n>f$ as, e.g., implemented in \cite[Algorithm 1]{DBLP:conf/sac/CorreiaVL10}: it is a `single echo' algorithm with USIG-signed messages and attached counter value. Correct processes relay a message once and reject messages with invalid USIG signatures or with counter values already received which prevents equivocating messages for the same counter.
An instance of the Reliable Broadcast abstraction has two functions: $\mfn{broadcast}{r,m}$ to reliably broadcast exactly one arbitrary message $m$ for round $r$ to all processes in $P$, and $\mfn{delivered}{}$ which returns all messages that were received by the instance since the last call to $\mfn{delivered}{}$. 
We expect the Reliable Broadcast abstraction to fullfil the following properties:

\begin{definition}[Reliable Broadcast]
\label{def:rbc}
A sender $p_s \in P, n := |P|$ can call $\mfn{broadcast}{m}$. Correct processes deliver $(c, m)$ where $c \in \mathbb{N}_0$ and $m$ an arbitrary message satisfying the following properties: \\
    \textbf{RB-Agreement.} If a correct process $p_i \in P$ delivers $(c, m)$, then every other correct process $p_k \in P, k \not = i$ eventually delivers the same $(c,m)$.\\
    \textbf{RB-Integrity.} For each $c \in \mathbb{N}_0$, a correct process $p_k \in P$ delivers $(c, m)$ at most once.\\
    \textbf{RB-Validity.} If a correct sender calls $\mfn{broadcast}{m}$, then every correct process $p_i \in P$ eventually delivers $(c, m)$.
\end{definition}

Additionally, we assume an asynchronous common coin, e.g. as defined by Cachin et al.~\cite{DBLP:journals/joc/CachinKS05}, that produces a uniformly distributed common random number $p$ out of $\{p \mid p \in \mathbb{N}_0 \colon p < n \}$ for all correct processes and a name $i \in \mathbb{N}_0$ as soon as $f+1$ processes invoked $\mfn{toss}{i}$;  repetitive calls with same the $i$ yield the same $p$.
We further assume a trusted setup of the common coin and the USIGs using a public key infrastructure
(to set up the common coin's threshold signature scheme, dealerless variants \cite{DBLP:conf/pkc/Boldyreva03} and those with an asynchronous setup \cite{DBLP:conf/wdag/AbrahamAM10} exist).

\subsection{The Algorithm} 

\begin{figure}
    \centering
    \includegraphics[width=\textwidth]{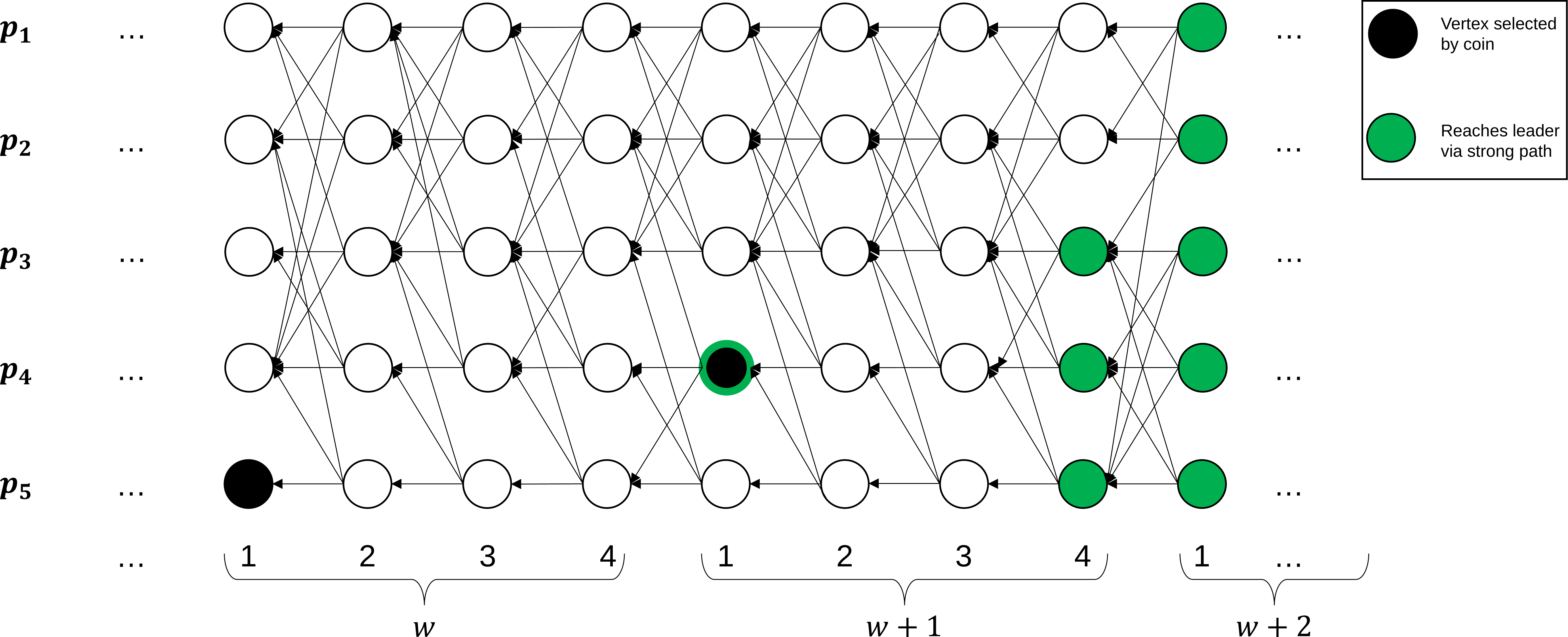}
    \caption{Example DAG for $n=5$ processes of which at maximum $2$ may be faulty. Shown is the `global' state of the graph, i.e., after every process eventually received every vertex. For simplicity, all vertices are valid and weak edges are left out. The direct commit rule is not fulfilled for any process for wave $w$; it is fulfilled for processes $p_3$, $p_4$, and $p_5$ for wave $w+1$. The green coloring highlights the effect of the direct commit rule as proven in \Cref{lemma:strongPathLeaders}. The direct commit rule ensures that a correct process can commit a wave retrospectively if it was not able to commit when it finished the wave. Since the leader vertex of wave $w+1$ has a strong path to the leader vertex of wave $w$, wave $w$ will be committed retrospectively.}
    \label{fig:graph}
\end{figure}

DAG-Rider uses $n$ Reliable Broadcast instances to disseminate process messages and to construct a  directed acyclic graph (DAG) that captures the communication history of all processes.
In a second step, each process derives consensus on the order of transactions using a Common Coin based on the graph structure.
The adapted DAG-Rider algorithm executed by a correct process $p_i \in P$ is shown in \Cref{alg:rider}.
Utility functions are listed in \Cref{alg:util}.
The core of the approach is the construction and interpretation of a (local) DAG that captures received transactions and the observed communication sequence between processes.
The DAG is structured in rounds and a round contains at maximum one vertex per process, i.e., $n$ vertices.
Rounds are addressed in an array style and the local view of a process $p_i$ on the DAG is indicated by an index $i$.
The very first round $\mathit{DAG}_i[0]$ is initialized with $n$ hard-coded ``genesis'' vertices.
A vertex in round $r$ has two types of edges: strong edges point to vertices of round $r-1$ and weak edges point to vertices of any round $r' \leq r-2$.
As soon as $p_i$ received $\quorum$ \emph{valid} vertices for a round $r$, i.e., $\quorum$ vertices referencing $\quorum$ vertices of round $r-1$ as strong edges ($v.\mathit{strong}$, l.~11), for which it also knows its predecessors (l.~15), $p_i$ will \emph{complete} round $r$ and transition to round $r+1$.
Now, as soon as $p_i$ receives a client transaction, it will become the payload of a vertex $v$ which is created and broadcast by $p_i$ for round $r+1$ (ll.~44 and 21-27). 
The vertex $v$ connects to all vertices $p_i$ received for round $r$ (l.~23).
If $p_i$ received vertices $u$ for older rounds that are not reachable from the newly created vertex using the transitive closure of strong and weak edges (a `path'), $u$ will become a weak edge of $v$ ($v.\mathit{weak}$, l.~26).
The new vertex is broadcast using Reliable Broadcast instance $i$ to all processes (l.~27).
Every fourth round a so-called wave, consisting of four rounds, is completed (l.~19) and the DAG structure is used to derive a total order on the transactions (ll. 28-42).
Each wave $w$ has exactly one wave leader $v$ which is chosen calling $\mathit{coin}.\mfn{toss}{w}$ from the vertices of $w$'s first $\mfn{round}{w,1}$.
The random number is used to select the process whose vertex is to be used as wave leader.
If $v$ was not (yet) received or there are no $\quorum$ vertices in the $w$'s fourth $\mfn{round}{w,4}$ that have $v$ in their transitive closure of strong edges (a `strong path'), i.e. the \emph{direct commit rule} is not fulfilled, the wave cannot be committed (l.~30).
If wave $w$ can be committed, process $p_i$ checks first if there are wave leaders of waves $w'$ between the last wave that was committed (variable $\mathit{decidedWave}$) and the current wave $w$ that were received in the meantime and are connected to the leader of the wave $w'+1$ (retrospective commit, ll.~33-36). 
The wave leaders are used as the root for a deterministic graph traversal to determine the total order of transactions (ll.~38-42).
An example for a resulting graph with $n=5$ processes, i.e. $f \leq 2$, is shown in \Cref{fig:graph}.

\begin{algorithm}
\caption{TEE-Rider pseudocode for process $p_i \in P, n := |P|, \boldsymbol{n \geq 2f+1}$}
\label{alg:rider}
\begin{algorithmic}[1]
\ReplicaState{$\mathit{DAG} \colon$ array of sets of vertices, $\mathit{DAG}[0]$ initialized with ``genesis'' vertices}
\ReplicaState{$r \colon \mathbb{N}$, initialized with 1}
\ReplicaState{$\mathit{decidedWave} \colon \mathbb{N}_0$, initialized with 0} 
\ReplicaState{$\mathit{transactionsToPropose} \colon$ queue of client transactions $t$, initialized empty}
\ReplicaState{$\mathit{buffer} \colon$ set of vertices, initialized empty}
\ReplicaState{$\mathit{rb} \colon$ array of $n$ Reliable Broadcast instances with delivered() and broadcast($r, m$)}
\ReplicaState{$\mathit{coin} \colon$ common coin instance with toss($w$)}

\While{True}
    \For{$k \gets 0$ up to $n-1$}
        \For{$m = (r', v) \in \mathit{rb}[k].\mfn{delivered}{}$}
            \If{$|v.strong| < \quorum$} \textbf{continue} 
            \EndIf
            \State $v.\mathit{source} \gets p_k; v.\mathit{round} \gets r'; v.\mathit{delivered} \gets \mathrm{False}$
            \State $\mathit{buffer}.\mfn{add}{v}$
        \EndFor
    \EndFor
    \For{$v \in \mathit{buffer}$}
        \If{$v.\mathit{round} > r \lor \exists u \in v.\mathit{strong} \cup v.\mathit{weak} \colon u \not \in \cup_{r'\geq 0}\mathit{DAG}[r']$}
            \textbf{continue}
        \EndIf
        \State $\mathit{DAG}[v.\mathit{round}].\mfn{add}{v}$
        \State $\mathit{buffer}.\mfn{remove}{v}$
    \EndFor
    \If{$|\mathit{DAG}[r]| < \quorum$} \textbf{continue} 
    \EndIf
    \If{$r \mod 4 = 0$}
        $\mfn{waveReady}{\frac{r}{4}}$
    \EndIf
    \State $r \gets r +1$
    \State \textbf{wait until} $\neg \mathit{transactionsToPropose}.\mfn{isEmpty}$
    \State $v \gets$ new vertex
    \State $v.\mathit{block} \gets \mathit{transactionsToPropose}.\mfn{dequeue}{}; v.\mathit{strong} \gets \mathit{DAG}[r - 1]$
    \For{$r' \gets r - 2$ down to $1$}
        \For{$u \in \mathit{DAG}[r']$}
            \If{$\neg \mfn{path}{v, u}$}
                $v.\mathit{weak}.\mfn{add}{u}$
            \EndIf
        \EndFor
    \EndFor
    \State $\mathit{rb}[i].\mfn{broadcast}{r, v}$
\EndWhile

\Proc{waveReady}{$w$}
    \State $v \gets \mfn{getWaveLeader}{w}$ \Comment{$\bot$ if $\mfn{round}{w,1}$ vertex of chosen process is not in $\mathit{DAG}$}
    \If{$v = \bot \lor |\{u \mid u \in \mathit{DAG}[\mfn{round}{w, 4}] \colon \mfn{strongPath}{u, v}\}| < \quorum$} \Return  
    \EndIf
    \State $\mathit{leadersStack} \gets$ new stack; $\mathit{leadersStack}.\mfn{push}{v}$
    \For{$w' \gets w - 1$ down to $\mathit{decidedWave} + 1$}
        \State $u \gets \mfn{getWaveLeader}{w'}$
        \If{$u \not = \bot \land \mfn{strongPath}{v, u}$}
            \State $\mathit{leadersStack}.\mfn{push}{u}; v \gets u$
        \EndIf
    \EndFor
    \State $\mathit{decidedWave} \gets w$
    \While{$\neg \mathit{leadersStack}.\mfn{isEmpty}{}$}
        \State $v \gets \mathit{leadersStack}.\mfn{pop}{}$
        \State $\mathit{verticesToDeliver} \gets \{u \mid u \in \cup_{r'>0} \mathit{DAG}[r'] \colon \mfn{path}{v, u} \land \neg u.\mathit{delivered} \}$
        \For{$u \in \mathit{verticesToDeliver}$ in deterministic order}
            \State $u.\mathit{delivered} \gets \mathrm{True}$; \textbf{deliver} $(u.\mathit{block}, u.\mathit{round}, u.\mathit{source})$
        \EndFor
    \EndWhile
\EndProc

\Upon{clientRequest}{$t$}
    \State $\mathit{transactionsToPropose}.\mfn{enqueue}{t}$
\EndUpon
\end{algorithmic}
\end{algorithm}

\begin{algorithm}
\caption{Utility functions pseudocode}
\label{alg:util}
\begin{algorithmic}[1]
\Fn{path}{$v,u$}{boolean}
    \State \Return exists a sequence of vertices $(v_1, v_2, ..., v_k) \in \cup_{r'\geq 0} DAG[r']$ such that
    \State $v_1 = v \land v_k = u \land \forall i \in [2,k] \colon v_i \in v_{i-1}.strong \cup v_{i-1}.weak$
\EndFn
\Fn{strongPath}{$v,u$}{boolean}
    \State \Return exists a sequence of vertices $(v_1, v_2, ..., v_k) \in \cup_{r'\geq 0} DAG[r']$ such that
    \State $v_1 = v \land v_k = u \land \forall i \in [2,k] \colon v_i \in v_{i-1}.strong$
\EndFn
\Fn{getWaveLeader}{$w$}{vertex or $\bot$}
    \State $p_j \gets coin.\mfn{toss}{w}$
    \If{$\exists v \in DAG[\mfn{round}{w, 1}] \colon v.source = p_j$}
        \Return $v$
    \EndIf
    \State \Return $\bot$
\EndFn
\Fn{round}{$w,i$}{$\mathbb{N}$}
    \State \Return $4(w-1)+i$
\EndFn
\end{algorithmic}
\end{algorithm}

\section{Correctness Argument}
\label{sec:proof}
Lemmas 1 and 2 of the original DAG-Rider publication \cite{DBLP:conf/podc/KeidarKNS21} are crucial for Total Order and Agreement and rely on quorum intersection arguments.
The following Lemmas \ref{lemma:strongPathLeaders} and \ref{lemma:get-core} show the corresponding results for a quorum size of $\quorum$.
Results for Integrity and Validity simply follow from the original paper.

\begin{lemma}
\label{lemma:strongPathLeaders}
If a correct process $p_i \in P$ commits the wave leader $v$ of a wave $w$ when it completes wave $w$ in $\mfn{round}{w,4}$, then any valid vertex $v'$ of any process $p_j \in P$ broadcast for a round $r \geq \mfn{round}{w+1,1}$ will have a strong path to $v$.
\end{lemma}
\begin{proof}
Since $p_i$ commits $v$ in $\mfn{round}{w,4}$, the direct commit rule is fulfilled (l.~30): $\exists U \subseteq \mathit{DAG}_i[\mfn{round}{w,4}]\colon |U| \geq \quorum \land \forall u \in U \colon \mfn{strongPath}{u,v}$.
A valid vertex must reference at least $\quorum$ distinct vertices of the previous round with a strong edge (l.~11).
Thus, a process $p_j \in P$ broadcasting a valid vertex $v_j$ for $\mfn{round}{w+1,1}$ selected at least $\quorum$ vertices of $\mfn{round}{w,4}$ as strong edges for $v_j$.
Any two subsets of size $\quorum$ of a superset of size $n$ intersect at least in one element.
Thus, every valid vertex of a process broadcast for $\mfn{round}{w+1,1}$ must have at least one edge to a vertex of $U$, and, via $U$ to $v$. 
As every valid vertex of $\mfn{round}{w+1,1}$ has a strong path to $v$, and every valid vertex of $\mfn{round}{w+1,2}$ connects to at least $\quorum$ vertices of  $\mfn{round}{w+1,1}$, by induction, any valid vertex $v'$ of any process $p_j \in P$ broadcast for a round $r \geq \mfn{round}{w+1, 1}$ has a strong path to $v$.
\qedhere
\end{proof}

\begin{lemma}
\label{lemma:get-core}
When a correct process $p_i \in P$ completes $\mfn{round}{w, 4}$ of wave $w$, then $\exists V_1 \subseteq \mathit{DAG}_i[\mfn{round}{w, 1}], V_4 \subseteq \mathit{DAG}_i[\mfn{round}{w, 4}] \colon |V_1| \geq \quorum \land |V_4| \geq \quorum \land (\forall v_1 \in V_1, \forall v_4 \in V_4 \colon \mfn{strongPath}{v_4,v_1})$.
\end{lemma}
\begin{proof}
By use of Reliable Broadcast and validity checks in ll.~11 and 15, faulty processes are limited to omission faults.
Thus, the \emph{get-core} argument of Attiya and Welch \cite[Sec. 14.3.1]{attiya2004distributed} still holds \cite[Sec. 14.3.3]{attiya2004distributed}:
Let $A \in \{0,1\}^{n \times n}$ be a matrix that contains a row for each possible vertex of $\mfn{round}{w,3}$ and a column for each possible vertex of $\mfn{round}{w,2}$. 
Let $A[j,k] = 1$ if the vertex of process $p_j$ of $\mfn{round}{w,3}$ has a strong edge to the vertex of process $p_k$ of $\mfn{round}{w,2}$ or $p_j$ sends no vertex (or an invalid one) but $p_k$ sends a valid vertex for $\mfn{round}{w,2}$.
As there are at least $\quorum \leq n-f$ correct processes, each row of $A$ contains at least $\quorum$ ones and $A$ contains at least $n(\quorum)$ ones. 
Since there are $n$ columns, there must be a column $l$ with at least $\quorum$ ones. 
This implies there is a vertex $v_l$ by process $p_l$ in $\mfn{round}{w,2}$ s.t. $\exists V_3 \subseteq \mathit{DAG}_i[\mfn{round}{w, 3}] \colon |V_3| \geq \quorum \land \forall v_3 \in V_3 \colon \mfn{strongPath}{v_3, v_l}$.
As at most $f$ vertices in $V_3$ belong to faulty processes that may commit send omission faults for $\mfn{round}{w,3}$ and $\quorum \geq f+1$, by quorum section at least one vertex of $V_3$ is received by any correct process $p_j \in P$ before it sends its vertex for $\mfn{round}{w,4}$.
Thus, every valid vertex in $\mathit{DAG}_i[\mfn{round}{w, 4}]$ has at least one strong edge to a vertex of $V_3$.
Since $v_l$ must be valid and thus has a strong edge to each vertex of a set $V_1 \subseteq \mathit{DAG}_i[\mfn{round}{w,1}], |V_1| \geq \quorum$, any valid vertex of rounds $r \geq \mfn{round}{w,4}$ has a strong path to every vertex, including $V_1$, reached by $v_l$ via strong paths. Please note that the construction of the set $V_1$ is valid for all correct processes that complete the wave and, thus, represents the `common core'.
\qedhere
\end{proof}

\section{Discussion and Conclusion}
\label{sec:conclusion}
DAG-Rider shows the power of causal order broadcast to implement consensus.
The adaption for TEEs preserves the simplicity of DAG-Rider while increasing its fault tolerance and reducing the communication effort (i.e., from `double echo` to `single echo` Reliable Broadcast).
The ease of adaption of DAG-Rider for TEEs make it a perfect textbook example for TEE-based ABAB.
Follow-up work to DAG-Rider, Tusk \cite{DBLP:conf/eurosys/DanezisKSS22}, addresses a major deployability issue, namely garbage collection, shortens the wave length, and replaces the underlying Reliable Broadcast with a communication scheme that leverages the graph structure to achieve linear communication complexity in the happy case.
Additionally, to the best of our knowledge, there exists no TEE-based, dealerless, and asynchronous common coin primitive.
In summary, investigating a TEE-based dealerless setup as well as transforming the follow-ups of DAG-Rider for empirical studies to investigate the assumed superiority of asynchronous TEE-based approaches, e.g., in comparison to MinBFT \cite{DBLP:journals/tc/VeroneseCBLV13}, is a promising line of research. 

\bibliography{references.bib}

\end{document}